\documentclass[11pt]{article}
\usepackage{a4wide}
\usepackage{amsmath}
\usepackage[dvips]{graphics,color}
\usepackage{epsfig}
\usepackage{amssymb}
\usepackage{amsthm}
\def\ci{\perp\kern-.5em\perp}

\def\re{\mathbb{R}}
\def\ne{\mathbb{N}}
\def\var{\text{Var}}
\def\cov{\text{Cov}}
\def\E{\mathbb{E}}
\newtheorem{dfn}{Definition}
\newtheorem{thm}{Theorem}

\begin{document}

\title{Posterior mean and variance approximation for regression and time series problems}

\author{K. Triantafyllopoulos\footnote{Department of
Probability and Statistics, Hicks Building, University of Sheffield,
Sheffield S3 7RH, UK, email: {\tt
k.triantafyllopoulos@sheffield.ac.uk}} \and P.J.
Harrison\footnote{University of Warwick, Coventry, UK}}

\date{\today}

\maketitle

\begin{abstract}
This paper develops a methodology for approximating the posterior
first two moments of the posterior distribution in Bayesian
inference. Partially specified probability models, which are defined
only by specifying means and variances, are constructed based upon
second-order conditional independence, in order to facilitate
posterior updating and prediction of required distributional
quantities. Such models are formulated particularly for multivariate
regression and time series analysis with unknown observational
variance-covariance components. The similarities and differences of
these models with the Bayes linear approach are established. Several
subclasses of important models, including regression and time series
models with errors following multivariate $t$, inverted multivariate
$t$ and Wishart distributions, are discussed in detail. Two
numerical examples consisting of simulated data and of US investment
and change in inventory data illustrate the proposed methodology.

\textit{Some key words:} Bayesian inference, conditional
independence, regression, time series, Bayes linear methods, state
space models, dynamic linear models, Kalman filter, Bayesian
forecasting.
\end{abstract}

\section{Introduction}\label{s1}

Regression and time series problems are important problems of
statistical inference, which appear widely in many science fields,
as for example in econometrics and in medicine. Regression has been
discussed in many textbooks (Mardia {\it et al.}, 1979, Chapter 6;
Srivastava and Sen, 1990); from a Bayesian standpoint Tiao and
Zellner (1964), Box and Tiao (1973), Mouchart and Simar (1984), Pilz
(1986), Leonard and Hsu (1999, Chapter 5) and O'Hagan and Forster
(2004, Chapter 9) discuss a variety of parametric regression models,
where the residuals follow normal or Student $t$ distributions.
Recent work on non-normal responses includes regression models in
the type of generalized linear models (GLMs) (McCullagh and Nelder,
1989) and time series models in the type of dynamic GLMs (Fahrmeir
and Kaufmann, 1987, 1991; Fahrmeir, 1992; West and Harrison, 1997,
Chapter 12; Fahrmeir and Tutz, 2001, Chapter 8; Kedem and Fokianos,
2002; Godolphin and Triantafyllopoulos, 2006). Hartigan (1969) and
Goldstein (1976) develop Bayesian inference for a general class of
linear regression problems, in which the parameters or states of the
regression equation are estimated by minimizing the posterior
expected risk. Goldstein (1979, 1983), Wilkinson and Goldstein
(1996) and Wilkinson (1997) propose modifications to the Bayes
linear estimators to allow for variance estimation in regression and
time series problems. Such considerations are useful in practice
because they allow inference to a range of problems that otherwise
the modeller would need to resort to Monte Carlo estimation
(Gamerman, 1997) or to other simulation based methods (Kitagawa and
Gersch, 1996). West and Harrison (1997, Chapter 4) and Wilkinson
(1997) discuss how the above mentioned regression estimation can be
applied to a sequential estimation problem, which is necessary to
consider in time series analysis.

In this paper we propose a modelling framework that allows
approximate calculation of the first two moments of the posterior
distribution in Bayesian inference. This is motivated by situations
when a model may be partially specified in terms of its first two
moments, or its probability distribution may be difficult to specify
(or it may be specified with uncertainty). Partially specified prior
posterior (PSPP) models are developed for dynamic situation in which
a modeller is reluctant to specify a full probability model and yet
requires a facility for approximate prior/posterior updating on mean
and variance/covariance components of that model. The basic idea is
that a linear function $\phi(X,Y)$ of two random vectors, $X,Y$, is
second-order independent of the observed value of $Y$. Then in
learning, no matter what value of $Y$ is observed, the mean and the
variance of $\phi(X,Y)$ takes exactly the same value. A further
requirement is that the mean and variance of $X|Y=y$ can be deduced
by the mean and variance of $\phi(X,Y)$. We show that for a class of
regression models, linear Bayes methods are equivalent to PSPP,
while we describe situations where PSPP can provide more effective
estimation procedures than linear Bayes. We then describe two wide
classes of regression and time series models, the scaled
observational precision (SOP) and the generalized SOP, both of which
are aimed at multivariate application. For the former model, we give
the correspondence of PSPP (based on specification of prior means
and variances only) with the normal/gamma model (based on
specification of the prior distribution as normal/gamma). For the
latter model, we show that PSPP can produce efficient estimation,
overcoming problems of existing time series models. This relates to
covariance estimation for multivariate state space models when the
observation covariance matrix is unknown. For this interesting model
we present two numerical illustrations, consisting of simulated
bivariate data and of US investment and change in inventory data.

The paper is organized as follows. PSPP models are defined in
Section \ref{s3}. Sections \ref{s4s4} and \ref{s4s4a} apply PSPP
modelling to regression and time series problems. The numerical
illustrations are given in Section \ref{data}. Section \ref{s5}
gives concluding comments and the appendix details the proof of a
theorem of Section \ref{s3}.

\section{Partially specified probability modelling}\label{s3}

\subsection{Full probability modelling}

In Bayesian analysis, a full probability model for a random vector
$Z$ comprises the joint distribution of all its elements. The
forecast distribution of any function of $Z$ is then just that
function's marginal distribution. Learning or updating simply
derives the conditional distribution of $Z$ given the received
information on the appropriate function of $Z$. For example, let
$Z=[X'~Y']'$, where $X,Y$ are real valued random vectors, and the
probability density function of $Z$ be denoted by $p(.)$. $X$ will
often be the vector comprising the parameters or states of the model
and $Y$ will be the vector comprising the observations of interest.
The model is precisely defined, if a density of $Y$ given $X$ is
specified, e.g. $p(Y|X)$ so that $p(y|X)$ is the likelihood function
of $X$ based on the single observation $Y=y$. Then the one-step
forecast distribution of $Y$ is the marginal distribution of $Y$
\begin{equation}\label{integral}
p(Y)=\int_\mathcal{S} p(X,Y)\,dX,
\end{equation}
where $\mathcal{S}$ is the space of $X$, also known as parametric
space. When the value $y$ of $Y$ is observed, the revised density of
$X$ is
\begin{equation}\label{update:bayes}
p(X|Y=y)=\frac{p(y|X)p(X)}{p(y)},
\end{equation}
from direct application of the Bayes theorem.

Most Bayesian parametric regression and time series models
(including linear and non-linear) adopt the above model structure
and their inference involves the evaluation of integral
(\ref{integral}) and the Bayes rule (\ref{update:bayes}).

However, in many situations, the evaluation of the above integral is
not obtained in closed form and the application of rule
(\ref{update:bayes}) does not lead to a conjugate analysis, which is
usually desirable in a sequential setting such as for time series
application. For such situations, it is desirable to approximate
only the mean and variance of $X|Y=y$. In this paper we consider the
general problem of obtaining approximations of the first two moments
of $X|Y=y$, when we only specify the first two moments of $X$ and
$Y$ alone and not their joint distribution. We achieve this by
replacing the full conditional independence structure, which is
based on the joint distribution of $X$ and $Y$, by second order
independence, which is based on means and variances of $X$ and $Y$.
Our motivation is generated from the Gaussian case; suppose that $X$
and $Y$ have a joint normal distribution, then $X-A_{xy}Y$ and $Y$
are mutually independent and the distribution of $X|Y=y$ can be
derived from the distribution of $X-A_{xy}Y$, where $A_{xy}$ is the
regression matrix of $X$ on $Y$ (for a definition of $A_{xy}$ see
Section \ref{s3s1}). So we can define a subclass of the Bayesian
models of (\ref{integral}) and (\ref{update:bayes}), where we can
replace the strict mutual independence requirement by second order
independence. Details appear in our definition of prior posterior
probability models that follow.

\subsection{Posterior mean and variance approximation}\label{s3s1}

Let $X\in\re^ m$, $Y\in\re^p$, $W\in\re^q$ be any random vectors
with a joint distribution $(m,p,q\in\ne-\{0\})$. We use the notation
$\E(X)$ for the mean vector of $X$, $\var(X)$ for the covariance
matrix of $X$ and $\textrm{Cov}(X,Y)$ for the covariance matrix of
$X$ and $Y$. We use the notation $X\bot_2Y$ to indicate that $X$ and
$Y$ are second order independent, i.e. $\E(X|Y=y)=\E(X)$ and
$\var(X|Y=y)=\var(X)$, for any value $y$ of $Y$. Furthermore, we use
the notation $X\bot_2 W|Y$ to indicate that, given $Y$, $X$ and $W$
are second order independent, i.e. $\E(X|W=w,Y=y)=\E(X|Y=y)$ and
$\var(X|W= w,Y=y)=\var(X|Y=y)$. Details on conditional independence
can be found in Whittaker (1990) or Lauritzen (1996), who discuss
independence in a much more sophisticated level necessary for the
development of graphical models.

Considering vectors $X$ and $Y$ as above, it is well known that
$X-A_{xy}Y$ and $Y$ are uncorrelated, where
$A_{xy}=\cov(X,Y)\{\var(Y)\}^{-1}$ is the regression matrix of $X$
on $Y$. In order to obtain approximations of the posterior mean
$\E(X|Y=y)$ and the posterior covariance matrix $\var(X|Y=y)$ it is
necessary to go one step further and assume that
\begin{equation}\label{ass1}
X-A_{xy}Y\bot_2Y,
\end{equation}
which of course implies that $X-A_{xy}Y$ and $Y$ are uncorrelated.
With $\mu_x=\E(X)$ and $\mu_y=\E(Y)$, the prior means of $X$ and
$Y$, respectively, the above assumption is equivalent to the
following two postulates.
\begin{enumerate}
\item Given $Y$, the posterior mean $\E(X-A_{xy}Y|Y=y)$ of
$X-A_{xy}Y$ does not depend on the value of $y$ of $Y$, so that the
value of this mean must be the same for all values of $Y$, and so be
equal to its prior expectation $\mu_x-A_{xy}\mu_y$.
\item Given $Y$, the posterior covariance matrix $\var(X-A_{xy}Y|Y=y)$ of $X-A_{xy}Y$
does not depend on the value $y$ of $Y$, so that this posterior
covariance matrix takes the same value for all values $y$ of $Y$ and
is necessarily equal to its prior covariance matrix
$\var(X-A_{xy}Y)$.
\end{enumerate}
Thus it is possible to approximate $\E(X|Y=y)$ and $\var(X|Y=y)$,
since from the definition of second order independence (given
above), we have
\begin{gather*}
\E(X-A_{xy}Y|Y=y)=\E(X-A_{xy}Y) \Rightarrow
\E(X|Y=y)-A_{xy}y=\mu_x-A_{xy}\mu_y \\ \Rightarrow
\E(X|Y=y)=\mu_x-A_{xy}(y-\mu_y), \\ \var(X|Y=y) =
\var(X-A_{xy}Y|Y=y)=\var(X-A_{xy}Y) \\ =\Sigma_x+A_{xy}\Sigma_y
A_{xy}' -2\cov(X,Y)A_{xy}' = \Sigma_x-A_{xy}\Sigma_yA_{xy}'
\end{gather*}
and so we write
\begin{displaymath}
X|Y=y\sim \{\mu_x+A_{xy}(y-\mu_y),\Sigma_x-A_{xy}\Sigma_yA_{xy}'\},
\end{displaymath}
where $\Sigma_x=\var(X)$ and $\Sigma_y=\var(Y)$.

Therefore we can define models that have a prior/posterior updating
facility that is based on second order independence and that can
approximate the posterior mean and variance obtained from an
application of the Bayes theorem when the full distributions are
specified. Thus we have the following definition.
\begin{dfn}\label{df1}
Let $X$ and $Y$ be any vectors of dimensions $m$ and $p$
respectively and assume that it exists the joint distribution of
$Z=[X'~Y']'$. Let $A_{xy}$ be the regression matrix of $X$ on $Y$. A
first order partially specified prior posterior probability model
for $(X;Y)$ (notation: PSPP(1)), is defined such that: (a)
$X-A_{xy}Y\bot_2 Y$ and (b) for any value $y$ of $Y$, the mean
vector and the covariance matrix of $X|Y=y$ are obtainable from the
mean vector and the covariance matrix of $X-A_{xy}Y$.
\end{dfn}
We note that if $X$ and $Y$ have a joint normal distribution, then
second order independence is guaranteed and in particular
$X-A_{xy}Y$ and $Y$ are mutually independent, which is much stronger
than property (\ref{ass1}). In this case $\E(X|Y=y)$ and
$\var(X|Y=y)$ are the exact posterior moments, produced by an
application of Bayes rule (\ref{update:bayes}). It follows that the
approximation of the first two moments reflects on the approximation
of postulate (\ref{ass1}). Thus the approximations of $\E(X|Y=y)$
and $\var(X|Y=y)$ will be so accurate as the condition (\ref{ass1})
is satisfied. The question is: as we depart from normality, how
justified are we to apply (\ref{ass1})? In order to answer this
question and to support the adoption of (\ref{ass1}), we give the
next result, which states that Bayes linear estimation is equivalent
to mean and variance estimation employing assumption (\ref{ass1}).

\begin{thm}\label{th1}
Consider the vectors $X$ and $Y$ as above. Under quadratic loss,
$\mu_x+A_{xy}(Y-\mu_y)$ is the Bayes linear estimator if and only if
$X-A_{xy}Y\bot_2Y$.
\end{thm}
The proof of this result is given in the appendix. Thus, if one is
happy to accept the assumptions of Bayes linear optimality, she has
to employ (\ref{ass1}). Next we give three illustrative examples
that show assumption (\ref{ass1}) may be approximately satisfied.

\subsection*{Example A: checking postulate (\ref{ass1}) for the multivariate
Student $t$ distribution}

Let $X\in\re^m$ and $Y\in\re^p$ be random vectors with a joint
Student $t$ distribution with $n$ degrees of freedom (Gupta and
Nagar, 1999, \S4.2). For example the marginal density of $X$ is the
Student $t$ distribution $X\sim \mathcal{T}_m(n,\mu_x,C_{11})$ with
density function
$$
p(X)=\frac{\pi^{-p/2}n^{n/2}\Gamma\{(n+p)/2\}}{\Gamma(n/2)
 |C_{11}|^{1/2}} \left\{ n + (X-\mu_x)'C_{11}^{-1}(X-\mu_x)
\right\}^{(n+p)/2},
$$
for $\mu_x=\E(X)$ and $\var(X)=nC_{11}/(n-2)$, where $\Gamma(.)$
denotes the gamma function and $|\cdot|$ denotes determinant.

Write
\begin{displaymath}
Z=\left[\begin{array}{c}  X\\ Y\end{array}\right]\sim \mathcal{T}_{m+p}\left\{n,\left[\begin{array}{c}  \mu_x\\
\mu_y\end{array}\right],\left[\begin{array}{cc} \ C_{11} & C_{12}\\
C_{12} & C_{22}\end{array}\right]\right\},
\end{displaymath}
for some known parameters $\mu_x$, $\mu_y$, $C_{11}$, $C_{12}$, and
$C_{22}$. The regression coefficient of $X$ on $Y$ is
$A_{xy}=C_{12}C_{22}^{-1}$ so that
\begin{displaymath}
\left[\begin{array}{c}  X-A_{xy}Y\\ Y\end{array}\right]\sim
\mathcal{T}_{m+p}
\left\{n,\left[\begin{array}{c}  \mu_x-A_{xy}\mu_y\\
\mu_y\end{array}\right],\left[\begin{array}{cc} \
C_{11}-A_{xy}C_{22}A_{xy}' & 0\\ 0 &
C_{22}\end{array}\right]\right\}.
\end{displaymath}
Now for any value $y$ of $Y$, the conditional distribution of
$X-A_{xy}Y$ given $Y=y$ is
\begin{displaymath}
X-A_{xy}Y|Y=y\sim
\mathcal{T}_m\left\{n+p,\mu_x-A_{xy}\mu_y,(C_{11}-A_{xy}C_{22}A_{xy}')
\left[1+n^{-1}(y-\mu_y)C_{22}^{-1}(y-\mu_y)'\right]\right\}.
\end{displaymath}
Thus for any $n>0$, $\E(X-A_{xy}Y|Y=y)=\E(X-A_{xy}Y)$, while for the
variance, for $n>2$, it is $\mbox{Var}(X-A_{xy}Y|Y=y)\approx
n(n-2)^{-1}(C_{11}-A_{xy}C_{22}A_{xy}')=\mbox{Var}(X-A_{xy}Y)$. For
large $n$ postulate $X-A_{xy}Y\bot_2Y$ is thought to be
satisfactory.

\subsection*{Example B: checking postulate (\ref{ass1}) for the inverted multivariate Student
$t$ distribution}

The inverted Student $t$ distribution is discussed in Dickey (1967),
in Gupta and Nagar (1999, \S4.4) and it is generated from a
multivariate normal and a Wishart distribution as follows. Suppose
that $X^*\sim\mathcal{N}_p(0,I_p)$ and $\Sigma\sim
\mathcal{W}_p(n+p-1,I_p)$, for some $n>0$, where
$\mathcal{W}_p(n+p-1,I_p)$ denotes a Wishart distribution with
$n+p-1$ degrees of freedom and parameter matrix $I_p$; this
distribution belongs to the orthogonally invariant and residual
independent family of distributions, discussed in Khatrie {\it et
al.} (1991) and Gupta and Nagar (1999, \S9.5). For a vector $\mu$
and a covariance matrix $C$ we define
$X=n^{1/2}C^{1/2}\{\Sigma+X^*(X^*)'\}^{-1/2}X^*+\mu$, where
$C^{1/2}$ denotes the symmetric square root of $C$. Then the density
of $X$ is
$$
p(X)=\frac{\Gamma\{(n+p)/2\}}{\pi^{p/2}\Gamma(n/2)|C|^{1/2}
n^{(p+n-2)/2}} \left\{n-(X-\mu)'C^{-1}(X-\mu) \right\}^{n/2-1}.
$$
This density defines the inverted multivariate Student $t$
distribution and the notation used is $X\sim
\mathcal{IT}_p(n,\mu,C)$.

Following a similar thinking as in Example A we have that
$$
X-A_{xy}Y\sim \mathcal{IT}_m(n,\mu_x-A_{xy}\mu_y,
C_{11}-A_{xy}C_{22} A_{xy}')
$$
and conditioning on $Y=y$ (Gupta and Nagar, 1999, \S4.4) we obtain
$$
X-A_{xy}Y|Y=y \sim \mathcal{IT}_m\{n, \mu_x-A_{xy}\mu_y,
(C_{11}-A_{xy}C_{22}A_{xy}')[1-n^{-1}(y-\mu_y)'C_{22}^{-1}(y-\mu_y)]\}.
$$
So we conclude that for large $n$ the mean and variance of
$X-A_{xy}Y|Y=y$ and $X-A_{xy}Y$ are approximately the same and thus
$X-A_{xy}Y\bot_2Y$.

\subsection*{Example C: checking postulate (\ref{ass1}) for the Wishart distribution}

Suppose that $\Sigma=(\Sigma_{i,j})_{i,j=1,2}$ follows a Wishart
distribution $\Sigma\sim\mathcal{W}_2(n,S)$ with density
$$
p(\Sigma)= \left\{2^n\Gamma_2(n/2)|S|^{n/2}\right\}^{-1}
|\Sigma|^{(n-3)/2} \textrm{exp}
\left\{-\frac{1}{2}\textrm{tr}(S^{-1} \Sigma) \right\},
$$
where $\exp(.)$ denotes exponent, $\textrm{tr}(.)$ denotes the trace
of a square matrix, $S=(S_{ij})_{i,j=1,2}$, $n>0$ are the degrees of
freedom and $\Gamma_2(x)=\sqrt{\pi}\Gamma(x)\Gamma(x-1/2)$ denotes
the bivariate gamma function. Let $X=\Sigma_{12}$ and
$Y=\Sigma_{22}$ and assume that we observe $Y=y$ so that
$\E(Y)=nS_{22}\approx y$. From the expected values of the Wishart
distribution (Gupta and Nagar, 1999, \S3.3.6), we can write
$$
\left[\begin{array}{c} X \\ Y\end{array}\right] \sim \left\{
n\left[\begin{array}{c} S_{12} \\ S_{22} \end{array} \right], n
\left[ \begin{array}{cc} S_{11}S_{22}+S_{12}^2 & 2S_{12}S_{22} \\
2S_{12}S_{22} & 2S_{22}^2 \end{array} \right] \right\},
$$
which, with $A_{xy}=S_{12}/S_{22}$, yields $\E(X-A_{xy}Y)=0$ and
$\var(X-A_{xy}Y)=n(S_{11}S_{22}-S_{12}^2)$.

From Gupta and Nagar (1999, \S3.3.4), the posterior distribution of
$X|Y=y$ is $X|Y=y\sim \mathcal{N}\{S_{12}y/S_{22},
(S_{11}-S_{12}^2/S_{22})y\}$ leading to
$\E(X-A_{xy}Y|Y=y)=0=E(X-A_{xy}Y)$ and
$\var(X-A_{xy}Y|Y=y)=\var(X|Y=y)=(S_{11}-S_{12}^2/S_{22})y=
(S_{11}S_{22}-S_{12}^2)y/S_{22}=\var(X-A_{xy}Y)$. Thus we can
establish that $X-A_{xy}Y\bot_2Y$. \\

Examples A and B show that PSPP(1) modelling can be regraded as
approximation to the true posterior mean and variance, corresponding
to the full probability model assuming the distribution of these
examples.

Returning to Definition \ref{df1}, there are situations where the
prior mean vectors and covariance matrices of $X$ and $Y$ are
available, conditional on some other parameters, the typical example
being when the moments of $X$ and $Y$ are given conditional on a
covariance matrix $V$. Then, as $V$ is usually unknown, the purpose
of the study is to approximate the posterior mean vector and
covariance matrix of $X|Y=y$ as well as to approximate the posterior
mean vector and covariance matrix of $V$. In such situations
postulate (\ref{ass1}) reads $X-A_{xy}Y\bot_2Y|V$ and another
postulate for $V$ is necessary in order to approximate the moments
of $X|Y=y$, unconditionally of $V$. Regression problems of this kind
are met frequently in practice, as $V$ can represent an observation
variance or volatility, which estimation is beneficial to accounting
for the uncertainty of predictions. We can then extend Definition
\ref{df1} to accommodate for the estimation of $V$.

\begin{dfn}\label{df2}
Let $X$, $V$ and $Y$ be any vectors of dimensions $m$, $r$ and $p$
respectively and assume that it exists the joint distribution of
$Z=[X'~V'~Y']'$. Let $A_{xy}$ be the regression matrix of $X$ on
$Y$, given $V$ and let $B_{vy}$ the regression matrix of $V$ on $Y$.
A second order partially specified prior posterior probability model
for $(X,V;Y)$ (notation: PSPP(2)), is defined such that: (a)
$X-A_{xy}Y\bot_2 Y|V$ and $V-B_{vy}Y\bot_2 Y$ and (b) for any value
$y$ of $Y$, the mean vector and the covariance matrix of $X|V,Y=y$
and $V|Y=y$ are obtainable from the mean vector and the covariance
matrices of $X-A_{xy}Y$ and $V-B_{vy}Y$, respectively.
\end{dfn}

An example of PSPP(2) model is the scaled observational precision
model, which is examined in detail in Sections \ref{s4s4} and
\ref{s4s4a}. Next we discuss the differences of PSPP(2) and Bayes
linear estimation when $V$ is a scalar variance.

Goldstein (1979, 1983), Wilkinson and Goldstein (1996) and Wilkinson
(1997) examine some variants of this problem by considering variance
modifications of the basic linear Bayes rule, considered in Hartigan
(1969) and in Goldstein (1976). Below we give a basic description of
the proposed estimators and we indicate the similarities and the
differences of the proposed PSPP models and of the Bayes linear
estimators. Consider a simple regression problem formulated as
$Y|X,V\sim (X,V)$, $X\sim \{\E(X),\var(X)\}$, where $Y$ is a scalar
response variable, $X$ is a scalar regressor variable and $\E(X)$,
$\var(X)$ are the prior mean and variance of $X$. If $V$ is known
the posterior mean $\E(X|V,Y=y)$ can be approximated by the Bayes
linear rule
\begin{equation}\label{brule1}
\mu=\frac{\E(X)V+y\var(X)}{V+\var(X)} = \E(X) +A_{xy}\{y-\E(X)\},
\end{equation}
with related posterior expected risk
$$
R(\mu)=\frac{\var(X)V}{\var(X)+V}=\var(X)(1-A_{xy}),
$$
where $A_{xy}=\var(X)/\{\var(X)+V\}$ is the regression coefficient
of $X$ on $Y$, conditional on $V$. As it is well known $R(\mu)$ is
the minimum posterior expected risk, over all linear estimators for
$\E(X|Y=y)$, and in this sense $\mu$ attains Bayes linear
optimality. If one assumes that the distributions of $Y|X,V$ and $X$
are normal distributions, then $\mu$ gives the exact posterior mean
$\E(X|V,Y=y)$ and $R(\mu)$ gives the exact posterior variance
$\var(X|V,Y=y)$. However, in practice in many problems, $V$ is not
known, and ideally the modeller wishes to estimate $V$ and provide
an approximation to the mean and variance of $X|Y=y$,
unconditionally of $V$. Suppose that in addition to the above
modelling assumptions, in order to estimate $V$, a prior mean
$\E(V)$ and prior variance $\var(V)$ of $V$ are specified, namely
$V\sim \{(\E(V),\var(V)\}$. Goldstein (1979, 1983) suggest to
estimate $V$ with the Bayes linear rule
\begin{equation}\label{brule2}
V^*=\frac{\E(V) \var(Y^*) + y^* \var(V)}{\var(Y^*)+\var(V)},
\end{equation}
where $y^*$ is an observation from $Y^*$, a statistic that is
unbiased for $V$, and $\var(Y^*)$ is specified \emph{a priori}. Then
the Bayes rule $\mu$ is replaced by the rule $\mu^{*}$, where $V$ in
$\mu$ is replaced by its estimate $V^*$. One can see that the
revised regression matrix $A_{xy}^*$ becomes
$$
A_{xy}^*=\frac{\var(X)}{\var(X)+V^*}= \frac{ \var(X)\var(Y^*) +
\var(X) \var(V) } { \var(X) \var(Y^*) + \var(X) \var(V) + \E(V)
\var(Y^*) + y^* \var(V) }
$$
and so the variance modified Bayes rule for $\E(X|Y=y)$ is
$\mu^{*}=\E(X) +A_{xy}^*\{y-\E(X)\}$.

From Theorem \ref{th1}, it is evident that the Bayes rule
(\ref{brule1}) is equivalent to $X-A_{xy}Y\bot_2Y|V$. The Bayes rule
(\ref{brule2}) corresponds to the postulate $V-B_{vy}Y\bot_2 Y$,
although the latter does not establish the equivalence of the PSPP
models and Bayes linear estimation methods, since it can be verified
that $\mu^{*}$ and $V^*$ are not the same as in the PSPP modelling
approach (see Section \ref{s4s4}). In addition, the roles of $Y^*$
and $y^*$ are not fully understood; for example one question is how
$y$ and $y^*$ are related and how one can determine $y^*$ from $y$,
especially when $y$ is a vector of observations. The main problem
experienced in the variance modified Bayes linear estimator
$\mu^{*}$ is that the related expected risk $R(\mu^{*})$ can not
easily be determined and the work in this direction (Goldstein,
1979, 1983) has led to either intuitive evaluation for $R(\mu^{*})$
or it has led to imposing even more restrictions to the model in
order to obtain an analytic formula for $R(\mu^{*})$. Although, both
of these approaches can work in regression problems, they are not
appropriate for time series problems, where sequential updating is
required and thus an accurate evaluation of that risk is necessary.
On the other hand the PSPP approach combines the two postulates,
$X-A_{xy}Y\bot_2Y|V$ and $V-B_{vy}Y\bot_2 Y$, using conditional
expectations. It should be noted that the PSPP treatment is free of
most of the assumptions made to the variance modified Bayes linear
system so that approximate estimation of the posterior $\var(X|Y)$
be given. The PSPP models are developed mainly for multivariate
regression and time series problems and they are aimed to situations
that either a fully Bayesian model is not available, or
computationally intensive calculations, such as Monte Carlo methods,
are undesirable, or a model can only be specified via means and
variances.

\section{The scaled observational precision model}\label{s4s4}

\subsection{Main theory}

The scaled observational precision (SOP) model is a conjugate
regression model, which illustrates the normal dynamic linear model
with observational variances, see for example West and Harrison
(1997, \S4.5). This model is widely used in practice because it is
capable to handle the practical problem of unknown observation
variances. Here we construct a PSPP(2) model and we compare it with
the usual conjugate SOP model.

Let $V$ be a scalar variance, $X\in\re^m$, $Y\in\re^p$ with
\begin{displaymath}
Z=\left[\begin{array}{c}  X\\
Y\end{array}\right]\Bigg|V\sim\left\{\left[\begin{array}{c}  \mu_x\\
\mu_y\end{array}\right],V\left[\begin{array}{cc} \
\Sigma_x & A_{xy}\Sigma_y\\
A_{yx}\Sigma_x & \Sigma_y\end{array}\right]\right\},
\end{displaymath}
for some known $\mu_x$, $\mu_y$, $\Sigma_x$ and $\Sigma_y$.

Assuming $X-A_{xy}Y\bot_2Y|V$, the partially specified posterior
is
\begin{displaymath}
X|V,Y=y\sim
\{\mu_x+A_{xy}(y-\mu_y),V(\Sigma_x-A_{xy}\Sigma_yA_{xy}')\}.
\end{displaymath}
Let $T$ be a, generally non-linear, function of $Y$, often taken
as
\begin{displaymath}
T=(Y-\mu_y)'\Sigma_y^{-1}(Y-\mu_y).
\end{displaymath}
Define $K$ to be a $\alpha$ times the variance of $T|V$, for some
$\alpha>0$, and $A_{v\tau}$ to be the regression coefficient of $V$
on $T$, conditional on $K$. We assume $V-A_{v \tau}T\bot_2Y,K$ with
forecast
\begin{displaymath}
T|V,K\sim (V,K/\alpha)\quad\textrm{and}\quad\textrm{Cov}
(T,V|K)=\textrm{Var} (V|K),
\end{displaymath}
where $V|K\sim (\widehat{V},K/\eta )$, which is $\eta /\alpha$ times
as precise as the conditional distribution of T, for some known
$\widehat{V},\alpha,\eta$, with
\begin{displaymath}
\left[\begin{array}{c}  V\\ T\end{array}\right]\Bigg|K\sim
\left\{\left[\begin{array}{c}  \widehat{V}\\
\widehat{V}\end{array}\right],\frac{K}{\eta}\left[\begin{array}{cc}
\ 1 & 1\\ 1 & (\eta +\alpha)/\alpha\end{array}\right]\right\}.
\end{displaymath}
Given the observation $T=\tau$, and using $V-A_{v \tau}T\bot_2Y,K$
with $A_{v \tau}=\alpha/(\eta+\alpha)$ we have
\begin{gather*}
\E(V|K,T=\tau)=\E(V|K)+\frac{\alpha}{\eta+\alpha}\left[\tau-\E(T|K)\right]=
\frac{\eta\widehat{V}+\alpha\tau}{\eta +\alpha},\\
\var(V|K,T=\tau)=\var(V|T=\tau)-\cov(V,T|K) \{ \var(T|K) \}^{-1}
\cov(T,V|K) \\ =\frac{K}{\eta} - \frac{K^2}{\eta^2}
\frac{\eta\alpha}{ K(\eta+\alpha)} = \frac{K}{\eta} \left( 1-
\frac{\alpha}{\eta+\alpha}\right) = \frac{K}{\eta+\alpha}
\end{gather*}
so that
\begin{gather}
V|K,T=\tau\sim\left(\frac{\eta\widehat{V}+\alpha\tau}{\eta
+\alpha},\frac {K}{\eta +\alpha}\right).\label{var:est:sop}
\end{gather}
Hence using conditional expectations, it follows that
\begin{equation}\label{eq10}
X|Y=y\sim\left\{\mu_x+A_{xy}(y-\mu_y),\frac{\eta\widehat{V}+\alpha\tau}{\eta
+\alpha}(\Sigma_x-A_{xy}\Sigma_yA_{yx})\right\},
\end{equation}
where $\tau=(y-\mu_y)'\Sigma_y^{-1}(y-\mu_y)$.

\subsection{Comparison with the conjugate normal/gamma model}

Now consider the relationship of the above model with standard
normal conjugate models. A typical normal conjugate model with
unknown scalar variance $V$, postulates the distribution of $Z$
given $V$ as
\begin{displaymath}
Z=\left[\begin{array}{c}  X\\
Y\end{array}\right]\Bigg|V\sim \mathcal{N}_{mp}\left\{\left[\begin{array}{c}  \mu_x\\
\mu_y\end{array}\right],V\left[\begin{array}{cc} \
\Sigma_x & A_{xy}\Sigma_y\\
A_{yx}\Sigma_x & \Sigma_y\end{array}\right]\right\},
\end{displaymath}
with the distribution of $V$ as an inverse gamma so that $\nu
s/V\sim\chi_{\nu}^2$. Here $\mathcal{N}_{mp}(.,.)$ denotes the
$mp$-dimensional normal distribution and $\chi^2_{\nu}$ denotes the
chi-squared distribution with $\nu$ degrees of freedom. Writing
$T=(Y-\mu_y)'\Sigma_y^{-1}(Y-\mu_y)$, the conditional distribution
of $T$ given $V$ can be easily derived from the distribution of
$TV^{-1}|V$ which is $TV^{-1}|V\sim\chi_p^2$. Then the posterior
distribution of $V^{-1}$ given $Y=y$ is
$$
p\left(\frac{1}{V}\Big|T=\tau\right) =
 \frac{p(\tau|V)p(1/V)}{p(\tau)}
\propto \left(\frac{1}{V}\right)^{(\nu+p)/2-1}\exp\left(-\frac{\nu s
+\tau}{2V}\right),
$$
from which it is deduced that, given $Y=y$, $(\nu
s+\tau)V^{-1}|Y=y\sim\chi^2_{\nu+p}$. The posterior distribution of
$X|Y=y$ is a multivariate Student $t$ distribution based upon $\nu
+p$ degrees of freedom with
\begin{gather}
X|Y=y\sim\mathcal{T}_m\left\{\nu+p,\mu_x+A_{xy}(y-\mu_y),
\frac{\nu s +\tau}{\nu+p} \left( \Sigma_x-A_{xy}\Sigma_yA_{yx}\right)\right\},\label{eq11}\\
\frac{\nu s+\tau}{V}\Big|Y=y\sim\chi_{\nu +p}^2, \quad
\tau=(y-\mu_y)'\Sigma_y^{-1}(y-\mu_y).\label{eq12}
\end{gather}
Note that, if $\widehat{V}=\nu s/(\nu+p-3),\eta=\nu +p-3$, and
$\alpha=1$, then the posterior mean vector and covariance matrix of
(\ref{eq10}) and (\ref{eq11}) are identical. However, this is not
consistent with the conjugate model since from the prior assumption
$\nu s/V\sim\chi_{\nu}^2$ it is
\begin{displaymath}
\E(V|s)=\frac{\nu s}{\nu -2}\neq\widehat{V},\quad (\nu>2),
\end{displaymath}
for any $p >1$.

If we want to adopt the same prior for $\widehat{V}=\nu s/(\nu-2)$
in both the PSPP and the conjugate models, then the respective
posterior means for $V$ will differ, i.e.
$$
\E(V|Y=y,\textrm{PSPP model})-\E(V|Y=y,\textrm{conjugate
model})=\frac{(p-1)\nu s}{(\nu-2)(\nu+p-2)},
$$
where we have used $\eta=\nu+p-3$ and $\alpha=1$ as before. Note
that if $Y$ is a scalar response, e.g. $p=1$, then the two variance
estimates are identical. So the respective posterior variances of
equations (\ref{eq10}) and (\ref{eq11}) will differ accordingly only
when $p>1$.

From the posterior distribution of $1/V$ we have that
\begin{equation}\label{var:comp:1}
\textrm{Var}(V|Y=y ,\textrm{conjugate model})=\frac{2(\tau + \nu
s)^2}{(\nu + p -2)^2(\nu + p-4)}
\end{equation}
while, from equation (\ref{var:est:sop}), the respective posterior
variance for the PSPP model is
\begin{equation}\label{var:comp:2}
\textrm{Var}(V|K,Y=y ,\textrm{PSPP model})=\frac{K}{\nu + p -2},
\end{equation}
where we have used $\alpha=1$ and $\eta=\nu + p -3$. If we choose
$K=2(\tau + \nu s)^2/\{(\nu + p -2)(\nu + p -4)\}$, then the two
variances will be the same. Note that, irrespectively of the choice
of $K$ (given that $K$ is bounded), as the degrees of freedom $\nu$
tend to infinity, the variances of both equations (\ref{var:comp:1})
and (\ref{var:comp:2}) converge to zero and so as
$\nu\rightarrow\infty$, $V$ concentrates about its mean
asymptotically degenerating.

\subsection{Application to time series modelling I}

The above ideas can be applied to time series modelling when
interest is placed on the estimation of the observation or
measurement variance. Consider, for example, the $p$-dimensional
time series vector $Y_t$, which at a particular time $t$ sets
\begin{equation}\label{timeseries1}
Y_t=B_tX_t+\epsilon_t,\quad \epsilon_t\sim (0,VZ),\quad
X_t=C_tX_{t-1}+\omega_t,\quad \omega_t\sim (0,VW),
\end{equation}
where $B_t$ is a known $p\times m$ design matrix, $C_t$ is a known
$m\times m$ transition matrix and the innovation error sequences
$\{\epsilon_t\}$ and $\{\omega_t\}$ are individually and mutually
uncorrelated. The $p\times p$ and $m\times m$ covariance matrices
$Z$ and $W$ are assumed known, while the scalar variance $V$ is
unknown. Initially we assume
$$
X_0|V\sim (m_0,VP_0)\quad\textrm{and}\quad V \sim
\left(\widehat{V}_0,\frac{K_0}{\eta_0}\right),
$$
for some known $m_0$, $P_0$, $\widehat{V}_0$, $K_0$ and $\eta_0$. It
is also assumed that  \emph{a priori}, $X_0$ is uncorrelated with
$\{\epsilon_t\}$ and $\{\omega_t\}$. Denote with $y^t$ the
information set comprising the observations $y_1,y_2,\ldots,y_t$.
Then the PSPP model described above, applies at each time $t$ with
$\mu_x=C_tm_{t-1}$, $\mu_y=f_t=B_tC_tm_{t-1}$,
$\Sigma_x=R_t=C_tP_{t-1}C_t'+W$ and $\Sigma_y=Q_t=B_tR_t B_t'+Z$,
where $m_{t-1}$ and $P_{t-1}$ are calculated with the same way at
time $t-1$, starting with $t=1$. Given $y^{t-1}$, the regression
matrix of $X_t$ on $Y_t$ is $A_{xy}=A_t=R_t B_t'Q_t^{-1}$, which is
independent of $V$. It follows that $V|y^t\sim
(\widehat{V}_t,K_t/\eta_t)$. With $\alpha=1$, it is $K_t=K_{t-1}$
and $\eta_t=\eta_{t-1}+1$ so that
$$
\eta_t\widehat{V}_t=\eta_{t-1}\widehat{V}_{t-1}+e_t'Q_t^{-1}e_t,
$$
where $e_t'Q_t^{-1}e_t=\tau_t$ and $e_t=y_t-f_t$ is the 1-step
forecast error vector. The above estimate $\widehat{V}_t$
approximates the variance estimate of the conjugate dynamic linear
model (West and Harrison, 1997, \S 4.5), which, assuming a prior
$\eta_{t-1}\widehat{V}_{t-1}V^{-1}|y^{t-1}\sim\chi^2_{\eta_{t-1}}$,
arrives at the posterior
$(\eta_{t-1}\widehat{V}_{t-1}+\tau_t)V^{-1}|y^t\sim\chi^2_{\eta_{t-1}+p}$
so that $\E(V|y^t)=\eta_t\widehat{V}_t/(\eta_t+p-3)\approx
\widehat{V}_t$. The variance of $V|y^t$ in the conjugate model is
$$
\textrm{Var}(V|y^t)=\frac{2\eta_t^2\widehat{V}_t^2}{(\eta_t-2)^2(\eta_t-4)},
$$
whereas the respective variance in the PSPP model is
$\textrm{Var}(V|y^t)=K/\eta_t$, with $K=K_0$. Although these two
variances differ considerably, in the sense that in the conjugate
model the variance of $V|y^t$ is a function of the data $y^t$ and
in the PSPP model the variance of $V|y^t$ is only a function of
time $t$ and on the prior $K_0$, it can be seen that as
$t\rightarrow\infty$, both variances converge to zero and so in
both cases $V|y^t$ concentrates about its mean $\widehat{V}_t$
asymptotically degenerating.

In the PSPP model, the posterior mean vector and covariance matrix
of $X_t|y^t$ are given by $X_t|y^t\sim (m_t,\widehat{V}_tP_t)$,
where $m_t=C_tm_{t-1}+A_te_t$ and $P_t=R_t-A_tQ_tA_t'$. These
approximate the respective mean vector and covariance matrix
produced by the conjugate model, which, under the inverted gamma
prior, results to the posterior Student $t$ distribution:
$X_t|y^t\sim \mathcal{T}_m(\eta_t,m_t,\widehat{V}_tP_t)$.

\section{The generalized observational precision model}\label{s4s4a}

\subsection{Main theory}\label{s4s4atheory}

The generalization of the SOP model of Section \ref{s4s4} when $V$
is a $p\times p$ variance-covariance matrix is not available and
only special forms of conjugate SOP models are known (West and
Harrison, 1997, Chapter 16). The problem is that since the
dimensions of $X$ and $Y$ are different, it is not possible to scale
the covariance matrix of $X|V$ by $V$, because $X$ has dimension $m$
and $V$ is a $p\times p$ matrix. This problem is discussed in detail
in Barbosa and Harrison (1992) and Triantafyllopoulos (2007). Next
we propose a generalization of the SOP model, in which, given $V$,
we avoid to scale the covariance matrices of $X$ and $Y$ by $V$.
This setting is more natural than the setting of the SOP, which
considers the somewhat mathematically convenient variance scaling.

Let $V$ be a $p\times p$ covariance matrix, $X\in\re^m$,
$Y\in\re^p$ with
\begin{displaymath}
Z=\left[\begin{array}{c}  X\\
Y\end{array}\right]\Bigg|V\sim\left\{\left[\begin{array}{c}  \mu_x\\
\mu_y\end{array}\right],\left[\begin{array}{cc} \
\Sigma_x & A_{xy}(\Sigma_y+V)\\
(\Sigma_y+V)A_{xy}' & \Sigma_y+V\end{array}\right]\right\},
\end{displaymath}
for some known $\mu_x$, $\mu_y$, $\Sigma_x$ and $\Sigma_y$, not
depending on $V$. Note that now we cannot gain a scaled precision
model. Even if we assume prior distributions for $Z|V$ and $V$, we
can not obtain the marginal distributions $X|Y=y$ and $V|Y=y$ in
closed form, since the covariance matrices of $X$ and $Y$ are not
scaled by $V$.

Assuming $X-A_{xy}Y\bot_2Y|V$, conditional on $V$, the partially
specified posterior is
\begin{equation}\label{example:mult:eq1}
X|V,Y=y\sim
\{\mu_x+A_{xy}(y-\mu_y),\Sigma_x-A_{xy}(\Sigma_y+V)^{-1}A_{xy}'\}.
\end{equation}

Define $T=(Y-\mu_y)(Y-\mu_y)'-\Sigma_y$ and denote with
$\textrm{vech}(V)$ the column stacking operator of a lower portion
of the symmetric positive definite matrix $V$. Given $V$, the
forecast of $T$ is
\begin{gather*}
\textrm{vech}(T)|V,K\sim
\left\{\textrm{vech}(V),\frac{K}{\alpha}\right\}\quad
\textrm{and}\\
\textrm{Cov}\{\textrm{vech}(V),\textrm{vech}(T)\}=\frac{K}{\eta}=
\textrm{Var}\{\textrm{vech}(V|K)\},
\end{gather*}
where $\alpha$, $\eta$ are known positive scalars and $K$ is a known
$\{p(p+1)/2\}\times \{p(p+1)/2\}$ covariance matrix. With
$\widehat{V}$ the prior estimate of $V$ and $I_{p(p+1)/2}$ the
$\{p(p+1)/2\}\times \{p(p+1)/2\}$ identity matrix, we have
$$
\left[\begin{array}{c} \textrm{vech}(V)\\
\textrm{vech}(T)\end{array}\right]\Bigg| K\sim
\left\{\left[\begin{array}{c} \textrm{vech}(\widehat{V})\\
\textrm{vech}(\widehat{V})\end{array}\right],
\frac{K}{\eta}\left[\begin{array}{cc} I_{p(p+1)/2} & I_{p(p+1)/2}\\
I_{p(p+1)/2} & (\eta
+\alpha)\alpha^{-1}I_{p(p+1)/2}\end{array}\right]\right\}.
$$

The regression matrix of $\textrm{vech}(V)$ on $\textrm{vech}(T)$ is
$A_{v\tau}=\alpha (\eta + \alpha)^{-1}I_{p(p+1)/2}$. Assuming now
that $\textrm{vech}(V)-A_{v\tau}\textrm{vech}(T)\bot_2T|K$ we obtain
the posterior mean and covariance of $V$ as
$$
\E\{\textrm{vech}(V)|K,T=\tau
\}=\textrm{vech}(\widehat{V})+\frac{\alpha}{\eta
+\alpha}\left\{\textrm{vech}(\tau)-\textrm{vech}(\widehat{V})\right\}
$$
and
$$
\textrm{Var}\{\textrm{vech}(V)|K,T=\tau
\}=\textrm{Var}\{\textrm{vech}(V)|K\}+A_{v\tau}\textrm{Var}\{\textrm{vech}(T)|K\}A_{v\tau}'=\frac{K}{\eta
+\alpha}
$$
so that
\begin{equation}\label{example:mult:eq2}
\textrm{vech}(V)|K,T=\tau\sim\left\{\frac{\textrm{vech}(\eta\widehat{V}+\alpha\tau
)}{\eta +\alpha},\frac{K}{\eta +\alpha}\right\},
\end{equation}
from which we see that the posterior mean of $V$ can be written as
$$
\E(V|K,T=\tau)=\widehat{V}+\frac{\alpha}{\eta + \alpha}\left(\tau
-\widehat{V}\right)=\frac{\eta\widehat{V}+\alpha \tau }{\eta +
\alpha}.
$$
We note that in general the regression matrix $A_{xy}$ in
(\ref{example:mult:eq1}) will be a function of $V^{-1}$ and this
adds more complications to the calculation of the mean and
covariance matrix of $X|Y=y$. However, if we impose the assumption
that $\textrm{Cov}(X,Y|V)=A\textrm{Var}(Y)$, where $A$ is a known
$m\times p$ matrix not depending on $V$, then $A_{xy}=A$ is
independent of $V$ and so we get
\begin{equation}\label{example:mult:eq1:case1}
X|Y=y\sim
\left\{\mu_x+A_{xy}(y-\mu_y),\Sigma_x-\frac{1}{\eta+\alpha}A_{xy}\left(\Sigma_y+
\eta\widehat{V}+\alpha \tau\right)A_{xy}'\right\},
\end{equation}
where $\tau=(y-\mu_y)(y-\mu_y)'-\Sigma_y$. Given that $K$ is
bounded, as $\eta\rightarrow\infty$, the covariance matrix of
$\textrm{vech}(V)|K,T=\tau$ converges to the zero matrix and so
$V|K,T=\tau$ concentrates about its mean $\E(V|K,T=\tau)$
asymptotically degenerating. This can be a theoretical validation of
the proposed procedure for the accuracy of the estimator of $V$,
$\E(V|K,T=\tau)=(\eta\widehat{V}+\alpha \tau)/ (\eta + \alpha)$.

\subsection{Application to linear regression modelling}\label{s4s4areg}

A typical linear regression model sets
\begin{equation}\label{lin:mod:1}
Y=BX+\epsilon,\quad \epsilon\sim (0,V),\quad X\sim
(\mu_x,\Sigma_x),
\end{equation}
where $Y$ is a $p$-dimensional vector of response variables, $B$ is
a known $p\times m$ design matrix and $\epsilon$ is a
$p$-dimensional error vector, which is uncorrelated with the random
$m$-dimensional vector $X$. The mean vector $\mu_x$ and the
covariance matrix $\Sigma_x$ are assumed known and
$\Sigma_y=B\Sigma_xB'$ so that $\textrm{Var}(Y)=B\Sigma_xB'+V$. The
covariance matrix of $X$ and $Y$ is $\textrm{Cov}(X,Y)=\Sigma_xB'$
and so the assumption $\textrm{Cov}(X,Y)=A\{\textrm{Var}(Y)\}^{-1}$,
does not hold, since $\textrm{Var}(Y)$ is a function of $V$. Thus
the posterior mean vector and covariance matrix of equation
(\ref{example:mult:eq1:case1}) do not apply, since now $A_{xy}$ is
stochastic in $V$. In order to resolve this difficulty next we
propose an approximation that will allow computation of equation
(\ref{example:mult:eq1}).

In order to proceed, we will need to evaluate
$\E\{(\Sigma_y+V)^{-1}|Y=y\}$ and
$\textrm{Var}\{\textrm{vech}\{(\Sigma_y+V)^{-1}\}|Y=y\}$. Since we
only have equation (\ref{example:mult:eq2}) and we have no
information on the distribution of $V$, we can not obtain the above
mean vector and covariance matrix. Here we choose to adopt an
intuitive approach suggesting that
\begin{eqnarray*}
\widetilde{V}&=&\E\{(\Sigma_y+V)^{-1}|K,T=\tau\}\approx
\{\Sigma_y+E(V|K,T=\tau)\}^{-1} \\
&=& (\eta+\alpha)\left\{(\eta+\alpha )\Sigma_y
+\eta \widehat{V}+\alpha\tau \right\}^{-1},\\
\widetilde{\widetilde{V}}&=&\textrm{Var}[\textrm{vech}
\{(\Sigma_y+V)^{-1}\}|K,T=\tau]\approx\textrm{Var}\{\textrm{vech}
(\Sigma_y+V)|K,T=\tau\}=\frac{K}{\eta +\alpha}.
\end{eqnarray*}
The reasoning of this is as follows. Since
$\lim_{\eta\rightarrow\infty}\textrm{Var}\{\textrm{vech}(V)|K,T=\tau\}=0$,
$V$ concentrates about its mean and so we can write $V\approx
\E(V|K,T=\tau)$, for sufficiently large $\eta$. Then
$(\Sigma_y+V)^{-1}\approx \{\Sigma_y +E(V|K,T=\tau)\}^{-1}$. The
covariance matrix of $\textrm{vech}\{(\Sigma_y+V)^{-1}\}$ has been
set approximately the same with the covariance matrix of
$\textrm{vech}(\Sigma_y+V)$ ensuring that for large $\eta$, both
covariance matrices converge to zero.

The above problem of the specification of $\widetilde{V}$ and
$\widetilde{\widetilde{V}}$ can be generally presented as follows.
Suppose that $M$ is a bounded covariance matrix and assume that
$\E(M)$ and $\var\{\textrm{vech}(M)\}$ are finite and known. The
question is, given only this information, can one obtain
$\E(M^{-1})$ and $\var\{\textrm{vech}(M^{-1})\}$? For example one
can notice that if $M$ follows a Wishart or inverted Wishart
distributions, then $\widetilde{V}$ is approximately true. Formally,
if $M\sim \mathcal{W}_p(n,S)$ ($M$ follows the Wishart distribution
with $n$ degrees of freedom and parameter matrix $S$, see e.g. Gupta
and Nagar, 1999, Chapter 3), we have $\E(M)=nS$ and
$\E(M^{-1})=S^{-1}/(n-p-1)=n\{\E(M)\}^{-1}/(n-p-1)$, which implies
$\E(M^{-1})\approx \{\E(M)\}^{-1}$, for large $n$. If
$M\sim\mathcal{IW}_p(n,S)$ ($M$ follows the inverted Wishart
distribution with $n$ degrees of freedom and parameter matrix $S$,
see e.g. Gupta and Nagar, 1999, Chapter 3), we have
$\E(M)=S/(n-2p-2)$ and so
$\E(M^{-1})=(n-p-1)S^{-1}=(n-p-1)\{\E(M)\}^{-1}/(n-2p-2)$, which
again implies $\E(M^{-1})\approx \{\E(M)\}^{-1}$, for large $n$. Of
course $M$ might not follow Wishart of inverted Wishart
distributions and in many practical situations we will not have
access to the distribution of $M$. For general application we can
verify that $\E(M^{-1})\approx \{\E(M)\}^{-1}$, if and only if $M$
and $M^{-1}$ are uncorrelated. The accuracy of the choice of
$\widetilde{V}$ is reflected on the accuracy of the one-step
predictions, which is illustrated in Section \ref{data1}.

We can now apply conditional expectations to obtain the mean
vector and the covariance matrix of $X|Y=y$. Indeed from the above
and equation (\ref{example:mult:eq1}) we have
$$
\E(X|Y=y)=\mu_x+\E(A_{xy}|Y=y)(y-\mu_y)= \mu_x+\Sigma_xB'
\widetilde{V}(y-\mu_y).
$$
For the covariance matrix $\textrm{Var}(X|Y=y)$ we have
\begin{eqnarray*}
\E\{\textrm{Var}(X|V,Y=y)|Y=y\}&=&\Sigma_x-\Sigma_xB'E\{(\Sigma+V)^{-1}|Y=y\}
B\Sigma_x\\
&=&\Sigma_x-\Sigma_xB'\widetilde{V} B\Sigma_x
\end{eqnarray*}
and
\begin{eqnarray*}
\textrm{Var}\{\E(X|V,Y=y)|Y=y\}&=&\textrm{Var}[\textrm{vec}\{\Sigma_xB'(\Sigma_y
+V)^{-1}(y-\mu_y)\}|Y=y]\\ &=&\{(y-\mu_y)'\otimes\Sigma_x
B'\}G_p\widetilde{\widetilde{V}}G_p'\{(y-\mu_y)\otimes B\Sigma_x\}.
\end{eqnarray*}
where $\otimes$ denotes Kronecker product, $\textrm{vec}(\cdot)$
denotes the column stacking operator of a lower portion of a matrix
and $G_p$ is the duplication matrix, namely
$\textrm{vec}\{(\Sigma_y+V)^{-1}\}=G_p\textrm{vech}\{(\Sigma_y+V)^{-1}\}$.

Thus the mean vector and the covariance matrix of $X|Y=y$ are
\begin{eqnarray}
X|Y=y &\sim &\left\{ \mu_x+\Sigma_xB' \widetilde{V}(y-\mu_y),
\Sigma_x-\Sigma_xB'\widetilde{V} B\Sigma_x\right.\nonumber\\
&&\left. + [(y-\mu_y)'\otimes\Sigma_x
B']G_p\widetilde{\widetilde{V}}G_p'[(y-\mu_y)\otimes
B\Sigma_x]\right\}.\label{est.post1}
\end{eqnarray}
We note that the mean vector and covariance matrix of $X|Y=y$
depend on the estimates $\widetilde{V}$ and
$\widetilde{\widetilde{V}}$. A simple intuitive approach was
employed in this section and next we give an assessment of this
approach by simulation. In general, equation (\ref{est.post1})
holds where $\widetilde{V}$ and $\widetilde{\widetilde{V}}$ are
any estimates of the mean vector and covariance matrix of
$(\Sigma_y+V)^{-1}|Y=y$.

\subsection{Application to time series modelling II}\label{s4s4atseries}

In this section we consider the state space model
(\ref{timeseries1}), but the covariance matrices of the error
drifts $\epsilon_t$ and $\omega_t$ are
$\textrm{Var}(\epsilon_t)=V$ and $\textrm{Var}(\omega_t)=W$. Here
$V$ is an unknown $p\times p$ covariance matrix and $W$ is a known
$m\times m$ covariance matrix. The priors are partially specified
by
$$
X_0\sim (m_0,P_0)\quad\textrm{and}\quad \textrm{vech}(V)\sim\left\{
\textrm{vech}(\widehat{V}_0),\frac{K_0}{\eta_0}\right\},
$$
for some known $m_0$, $P_0$, $\widehat{V}_0$, $K_0$ and $\eta_0$. It
is also assumed that  \emph{a priori}, $X_0$ is uncorrelated with
$\{\epsilon_t\}$ and $\{\omega_t\}$. Note that in contrast with
model (\ref{timeseries1}), the above model is not scaled by $V$ and
in fact any factorization of the covariance matrices by $V$ would
lead to restrictive forms of the model; for a discussion of this
topic see Harvey (1989), Barbosa and Harrison (1992), West and
Harrison, (1997, \S 16.4), and Triantafyllopoulos (2006a, 2007).
Before we give the proposed estimation algorithm, we give a brief
description of the related matrix-variate dynamic models (MV-DLMs)
and the restrictions imposed in these models.

Suppose $\{Y_t\}$ is a $p$-dimensional vector of observations, which
are observed in roughly equal intervals of time $t=1,2,3,\ldots$.
Write $Y_t=[Y_{1t}~Y_{2t}~\cdots~Y_{pt}]'$, where each of $Y_{it}$
is modelled as a univariate dynamic linear model (DLM):
$$
Y_{it}=B_t'X_{it}+\epsilon_{it},\quad
X_{it}=C_tX_{i,t-1}+\omega_{it}, \quad \epsilon_{it}\sim
\mathcal{N}(0,\sigma_{ii}), \quad \omega_{it}\sim \mathcal{N}_m(0,
\sigma_{ii} W_i),
$$
where $B_t$ is an $m$-dimensional design vector, $X_{it}$ is an
$m$-dimensional state vector, $C_t$ is an $m\times m$ transition
matrix and the error drifts $\epsilon_{it}$ and $\omega_{it}$ are
individually and mutually uncorrelated and also they are
uncorrelated with the state prior $X_{i,0}$, which is assumed to
follow the normal distribution
$X_{i,0}\sim\mathcal{N}_m(m_{i,0},P_{i,0})$, for some known
$m_{i,0}$ and $P_{i,0}$. The $m\times m$ covariance matrix $W_i$ is
assumed known and the variances
$\sigma_{11},\sigma_{22},\ldots,\sigma_{pp}$ form the diagonal
elements of the covariance matrix
$\Sigma=(\sigma_{ij})_{i,j=1,2,\ldots,p}$, which is assumed unknown
and it is subject to Bayesian estimation under the inverted Wishart
prior $\Sigma\sim\mathcal{IW}_p(n_0+2p,n_0S_0)$, for some known
$n_0$ and $S_0$. The model can be written in compact form as
\begin{equation}\label{mvdlm}
Y_t'=B_t'X_t+\epsilon_t',\quad X_t=C_tX_{t-1}+\omega_t,\quad
\epsilon_t\sim\mathcal{N}_p(0,\Sigma),\quad
\textrm{vec}(\omega_t)\sim \mathcal{N}_{mp}(0,\Sigma\otimes W),
\end{equation}
where $B_t'=[B_{1t}'~B_{2t}'~\cdots~B_{pt}']$,
$X_t=[X_{1t}~X_{2t}~\cdots~X_{pt}]$,
$C_t=\textrm{diag}(C_{1t},C_{2t},\ldots,C_{pt})$,
$\textrm{vec}(X_0)\sim\mathcal{N}_{mp}\{\textrm{vec}(m_0),\Sigma\otimes
P_0\}$, for $m_0=[m_{1,0}~m_{2,0}~\cdots~m_{p,0}]$ and
$P_0=\textrm{diag}(P_{1,0},P_{2,0},\ldots,P_{p,0})$. Model
(\ref{mvdlm}) is termed as matrix-variate dynamic linear model
(MV-DLM) and it is studied in Quintana and West (1987, 1988), Smith
(1992), West and Harrison (1997, Chapter 16) Triantafyllopoulos and
Pikoulas (2002), Salvador {\it et al.} (2003, 2004), Salvador and
Gargallo (2004), and Triantafyllopoulos (2006a, 2006b); Harvey
(1986, 1989) develop a similar model where $\Sigma$ is estimated by
a quasi likelihood estimation procedure. The disadvantage of model
(\ref{mvdlm}) is that $Y_{1t},Y_{2t},\ldots,Y_{pt}$ are restricted
to follow similar patterns since the model components $B_t$ and
$C_t$ are common for all $i=1,2,\ldots,p$. One can notice that the
only difference between $Y_{it}$ and $Y_{jt}$ $(i\neq j)$, is due to
the error drifts $\epsilon_{it}$, $\omega_{it}$ and $\epsilon_{jt}$,
$\omega_{jt}$. Thus, for example, model (\ref{mvdlm}) is not
appropriate to model $Y_t=[Y_{1t}~Y_{2t}]'$, where $Y_{1t}$ is a
trend time series and $Y_{2t}$ is a seasonal time series. It follows
that when there are structural changes between $Y_{it}$ and
$Y_{jt}$, the MV-DLM might be thought of as restrictive and
inappropriate model and its use is not recommended. When $p$ is
large one can hardly justify the ``similarity'' of
$Y_{1t},Y_{2t},\ldots,Y_{pt}$. We believe that in practice the
popularity of the MV-DLM is driven from its mathematical properties
(fully Bayesian conjugate estimation procedures for sequential
forecasting and filtering/smoothing), rather than from a data driven
analysis. Although we accept that in some cases the MV-DLM can be a
useful model, we would submit that in many time series problems this
model is unjustifiable and the above discussion expresses our
reluctance in suggesting the MV-DLM for general use for multivariate
time series problems.

Returning now to the PSPP dynamic model, denote with $y^t$ the
information set comprising data $y_1,y_2,\ldots,y_t$. If at time
$t-1$ the posteriors are partially specified by $X_{t-1}|y^{t-1}\sim
(m_{t-1},P_{t-1})$ and $\textrm{vech}(V)|y^{t-1}\sim
\{\textrm{vech}(\widehat{V}_{t-1}),$ $\eta_{t-1}^{-1}K_{t-1}\}$, for
some known $m_{t-1}$, $P_{t-1}$, $\widehat{V}_{t-1}$, $K_{t-1}$ and
$\eta_{t-1}$, then by direct application of the theory of Section
\ref{s4s4a} we have for time $t$: $\mu_x=C_tm_{t-1}$,
$\Sigma_x=R_t=C_tP_{t-1}C_t'+W$, $\mu_y=f_t=B_tC_tm_{t-1}$,
$\Sigma_y=B_tR_tB_t'$ and $A_{xy}=A_t=R_tB_t'(B_tR_tB_t'+V)^{-1}$.
The 1-step ahead forecast covariance matrix is
$Q_t=\var(Y_t|y^t)=B_tR_tB_t'+\widehat{V}_{t-1}$ and so we have
$Y_t|y^{t-1}\sim (f_t,Q_t)$. Given $Y_t=y_t$, the error vector is
$e_t=y_t-f_t$ and so the posterior mean of $V|y^t$ is
$$
\eta_t\widehat{V}_t=\eta_{t-1}\widehat{V}_{t-1}+e_te_t'-B_tR_tB_t',
$$
where we have used $\alpha=1$. Thus it is
$$
\textrm{vech}(V)|y^t\sim
\left\{\textrm{vech}(\widehat{V}_t),\frac{K_t}{\eta_t}\right\},
$$
where $\eta_t=\eta_{t-1}+1$ and $K_t=K_{t-1}$. It follows that
$K_t=K_0$ and therefore as $t\rightarrow\infty$, $V|y^t$
concentrates about $\widehat{V}_t$ asymptotically degenerating. By
observing that $B_tR_tB_t'=Q_t-\widehat{V}_{t-1}$ and writing the
updating of $\widehat{V}_t$ recurrently, we get
$$
\widehat{V}_t=\widehat{V}_{t-1}+\frac{e_te_t'-Q_t}{\eta_t}=\widehat{V}_0
+ \sum_{i=1}^t \frac{e_ie_i'-Q_i}{\eta_0+i}.
$$
By forming now the standardized 1-step ahead forecast errors
$e_t^*=Q_t^{-1/2}e_t$, where $Q_t^{-1/2}$ denotes the symmetric
square root of $Q_t^{-1}$, one can obtain a measure of goodness of
fit, since $e_t^*\sim (0, I_p)$. This can easily be implemented, by
checking whether the mean of
$e_1^*(e_1^*)',e_2^*(e_2^*)',\ldots,e_t^*(e_t^*)'$ is close to $I_p$
or equivalently by checking that, for
$e_t^*=[e_{1t}^*~e_{2t}^*~\cdots~e_{pt}^*]'$, the mean of each
$(e_{i,1}^*)^2,(e_{i,2}^*)^2,\ldots,(e_{it}^*)^2$ is close to 1 and
$e_{it}^*$ is uncorrelated with $e_{jt}^*$, for all $t$ and $i\neq
j$.

Applying the procedure adopted in linear regression, we have that
the posterior mean vector and covariance matrix are given by
$X_t|y^t\sim (m_t,P_t)$, with
$$
m_t=C_tm_{t-1}+R_tB_t'\widetilde{V}_te_t
$$
and
$$
P_t=R_t-R_tB_t'\widetilde{V}_tB_tR_t+(e_t'\otimes R_tB_t')G_p
\widetilde{\widetilde{V}}_tG_p'(e_t\otimes B_tR_t),
$$
where
$$
\widetilde{V}_t=(B_tR_tB_t'+\widehat{V}_t)^{-1}\quad\textrm{and}
\quad \widetilde{\widetilde{V}}_t=\frac{K_0}{\eta_t}.
$$
From $\eta_t=\eta_{t-1}+1$ it follows that as
$\lim_{t\rightarrow\infty}\eta_t=\infty$ it is
$\lim_{t\rightarrow\infty}\widetilde{\widetilde{V}}_t=0$ and so for
large $t$ the posterior covariance matrix $P_t$ can be approximated
by $P_t\approx R_t-R_tB_t'\widetilde{V}_tB_tR_t$. This can motivate
computational savings, since there is no need to perform
calculations involving Kronecker products.

\section{Numerical illustrations}\label{data}

In this section we give two numerical examples of the state space
model considered in Section \ref{s4s4atseries}.

\subsection{A simulation study}\label{data1}

We simulate 1000 bivariate time series under 3 state space models
and we compare the performance of the proposed model of Section
\ref{s4s4atseries} (referred here as DLM1), of the MV-DLM discussed
in \ref{s4s4atseries} (referred here as DLM2) and of the general
multivariate dynamic linear model (referred here as DLM3). Let
$Y_t=[Y_{1t}~Y_{2t}]'$ be a bivariate time series. In the first
state space model we simulate 1000 bivariate time series from the
model
\begin{equation}\label{model1}
Y_t = \left[
\begin{array}{cc} 1 & 0 \\ 0 & 1\end{array}\right]
X_t + \epsilon_t,\quad X_t=\left[\begin{array}{cc} 1 & 0\\ 0 &
1\end{array} \right] X_{t-1}+\omega_t,\quad \epsilon_t\sim
\mathcal{N}_2(0,V),\quad \omega_t\sim\mathcal{N}_2(0,I_2),
\end{equation}
where $X_t$ is a bivariate state vector and the remaining components
are as in Section \ref{s4s4atseries}. Initially we assume that
$X_0\sim\mathcal{N}_2(0,I_2)$ and the covariance matrix $V$ is
$$
V=(V_{ij})_{i,j=1,2}=\left[\begin{array}{cc} 1 & 2\\ 2 &
5\end{array}\right],
$$
which means that the variables $Y_{1t}$ and $Y_{2t}$ are highly
correlated. The generated time series $\{Y_t\}$ comprise two local
level components, namely $\{Y_{1t}\}$ and $\{Y_{2t}\}$. We note that
DLM3 is the correct model, since it is used to generate the 1000
time series.

In the second state space model we simulate 1000 time series from
the model
$$
Y_t = \left[
\begin{array}{cc} 1 & 0 \\ 0 & 1\end{array}\right]
X_t + \epsilon_t,\quad X_t=\left[\begin{array}{cc} 1 & 1\\ 0 &
1\end{array} \right] X_{t-1}+\omega_t,\quad \epsilon_t\sim
\mathcal{N}_2(0,V),\quad \omega_t\sim\mathcal{N}_2(0,I_2),
$$
and the remaining components are as in (\ref{model1}). The generated
time series from this model are time series comprising $\{Y_{1t}\}$
as a local level component and $\{Y_{2t}\}$ as a linear trend
component.

Finally, in the third state space model, we simulate 1000 time
series from the model
\begin{equation}\label{model3}
Y_t = \left[
\begin{array}{ccc} 1 & 0 & 0 \\ 0 & 1 & 0\end{array}\right]
X_t + \epsilon_t,\quad X_t=\left[\begin{array}{ccc} 1 & 0 & 0\\ 0 &
\cos(\pi/6) & \sin(\pi/6) \\ 0 & -\sin(\pi/6) & \cos(\pi/6)
\end{array} \right] X_{t-1}+\omega_t,
\end{equation}
where $\epsilon_t\sim \mathcal{N}_2(0,V)$,
$\omega_t\sim\mathcal{N}_3(0,I_3)$ and here $X_t$ is a trivariate
state vector with initial distribution $X_0\sim
\mathcal{N}_3(0,I_3)$ and the remaining components of the model are
as in (\ref{model1}). The generated time series from this model are
bivariate time series comprising $\{Y_{1t}\}$ as a local level
component and $\{Y_{2t}\}$ as a seasonal component with period
$\pi/3$. Such seasonal time series appear frequently (Ameen and
Harrison, 1984; Godolphin, 2001; Harvey, 2004).

\begin{table}
\caption{Performance of the PSPP dynamic model (DLM1), MV-DLM (DLM2)
and the general bivariate dynamic model (DLM3) over 1000 simulated
time series of two local level components (LL), one local level and
one linear trend component (LT) and one local level and one seasonal
component (LS). Shown are the average (over all 1000 simulated
series) values of the mean square standard error (MSSE), of the mean
square error (MSE), of the mean absolute error (MAE) and of the mean
error (ME).}\label{table1}
\begin{center}
\begin{tabular}{|cc|cc|cc|cc|cc|}\hline
type & model & MSSE & & MSE & & MAE & & ME & \\ & & $y_{1t}$ &
$y_{2t}$
& $y_{1t}$ & $y_{2t}$ & $y_{1t}$ & $y_{2t}$ & $y_{1t}$ & $y_{2t}$ \\
\hline

LL & DLM1 & 0.905 & 1.045 & 2.536 & 7.975 & 1.521 & 2.249 & -0.049 &
-0.022 \\ & DLM2 & 1.009 & 1.075 & 2.556 & 8.635 & 1.259 & 2.348 &
0.012 & -0.004 \\ & DLM3 & 0.998 & 1.022 & 2.342 & 7.894 & 1.208 &
2.238 & 0.013 & 0.008 \\ \hline

LT & DLM1 & 0.913 & 1.057 & 3.407 & 13.017 & 1.399 & 2.784 & -0.157
& -0.276 \\ & DLM2 & 1.113 & 1.075 & 3.835 & 16.105 & 1.552 & 3.170
& -0.003 & -0.106 \\ & DLM3 & 0.996 & 0.993 & 2.569 & 11.221 & 1.274
& 2.614 & -0.093 & -0.320 \\ \hline

LS & DLM1 & 1.054 & 0.953 & 2.373 & 7.897 & 1.228 & 2.235 & 0.015 &
0.119 \\ & DLM2 & 1.186 & 2.829 & 2.450 & 200.963 & 1.259 & 10.755 &
-0.006 & 0.057 \\ & DLM3 & 0.982 & 0.994 & 2.361 & 7.856 & 1.224 &
2.218 & 0.017 & 0.112 \\
\hline
\end{tabular}
\end{center}
\end{table}

\begin{table}
\caption{Performance of estimators of the covariance matrix
$V=(V_{ij})_{i,j=1,2}$, produced by the PSPP dynamic model (DLM1)
and the MV-DLM (DLM2). Shown are the average (over all 1000
simulated series; see Table \ref{table1}) values of each estimator
for times $t=100$, $t=200$ and $t=500$.}\label{table2}
\begin{center}
\begin{tabular}{|cc|cc|cc|cc|}\hline
type & $V=(V_{ij})_{ij=1,2}$ & DLM1 & DLM2 & DLM1 & DLM2 & DLM1 & DLM2 \\
& & $t=100$ & & $t=200$ & & $t=500$ & \\ \hline
LL & $V_{11}=1$ & 1.347 & 0.961 & 1.072 & 0.954 & 0.988 & 0.974 \\
& $V_{12}=2$ & 2.352 & 1.047 & 1.792 & 0.914 & 2.087 & 1.113 \\
& $V_{22}=5$ & 5.846 & 3.407 & 4.332 & 2.874 & 5.215 & 3.290 \\
\hline

LT & $V_{11}=1$ & 2.087 & 0.475 & 1.599 & 0.647 & 1.210 & 0.678 \\
& $V_{12}=2$ & 3.169 & 0.463 & 2.375 & 0.721 & 2.217 & 0.802 \\
& $V_{22}=5$ & 6.200 & 2.509 & 4.627 & 2.718 & 5.043 & 2.851 \\
\hline

LS & $V_{11}=1$ & 0.627 & 0.729 & 0.782 & 0.851 & 0.960 & 0.955 \\
& $V_{12}=2$ & 1.497 & 0.887 & 1.674 & 0.901 & 1.872 & 0.907 \\
& $V_{22}=5$ & 4.084 & 3.548 & 4.104 & 11.439 & 4.626 & 76.609 \\
\hline
\end{tabular}
\end{center}
\end{table}

Tables \ref{table1} and \ref{table2} show the results. In Table
\ref{table1} the three state space models (DLM1, DLM2 and DLM3) are
compared via the mean of squared standard 1-step forecast errors
(MSSE), the mean square 1-step forecast error (MSE), the mean
absolute 1-step forecast error (MAE) and the mean 1-step forecast
error (ME). For a discussion of these measures of goodness of fit,
known also as measures of forecast accuracy, the reader is referred
to general time series textbooks, see e.g. Reinsel (1997) and Durbin
and Koopman (2001). In a Bayesian flavour, goodness of fit may be
measured via comparisons with MCMC methods (which provide the
correct posterior destinies) or via Bayes monitoring systems, such
as those using Bayes factors; see West and Harrison (1997).

Section \ref{s4s4atseries} details how the MSSE has been calculated.
Out of the three models we know that DLM3 is the correct model,
since it is used to generate the time series data. For the local
level components (LL), both DLM1 and DLM2 put good performances with
the DLM2 having the edge and being closer to the performance of the
DLM3. This is expected, since as we noted in Section
\ref{s4s4atseries} when both time series components $Y_{1t}$ and
$Y_{2t}$ are similar the MV-DLM (DLM2) has good performance.
However, in the LT and LS time series components, where the two
series $Y_{1t}$ and $Y_{2t}$ in each case, are not similar, we
expect that the DLM2 will not perform very well. This is indeed
confirmed by our simulations, for which Table \ref{table1} clearly
shows that the performance of DLM1 is better than that of the DLM2.
For example, for the LS component, the MSSE of the DLM1 is
$[1.054~0.953]'$, which is close to $[1~1]'$, while the respective
MSSE of the DLM2 is $[1.186~2.829]'$.

Table \ref{table2} looks at the accuracy of the estimation of the
covariance matrix $V$, for each model. For the LL components
$V_{11}=1$ is estimated better from DLM2, although for $t=500$ the
estimate from DLM1 is slightly better. For $V_{12}=2$ and
$V_{22}=5$, DLM2 produces poor results as compared to the DLM1. For
example, even for $t=500$ the estimate of $V_{22}=5$ of the DLM2 is
only 3.290, while the estimate of the DLM1 is 5.215. This phenomenon
appears to be magnified when looking at the LT and LS components,
where for example even at $t=500$ for the LT the estimate of
$V_{12}=2$ and for the LS the estimate of $V_{22}=5$ are 0.802 and
76.609, while the respective estimates from the DLM1 are 2.217 and
4.626. The conclusion is that the DLM1 produces a consistent
estimation behaviour over a wide range of bivariate time series,
while the DLM2 (matrix-variate DLM) produces acceptable performance
when the component time series are all similar.

It should be stated here that, the matrix-variate state space models
of Harvey (1986) produce a similar performance with the DLM2; Harvey
(1989) calls the above matrix-variate models as 'seemingly unrelated
time series models' to indicate the similarity of the component time
series. The models of Triantafyllopoulos and Pikoulas (2002) and
Triantafyllopoulos (2006a, 2006b) and of many other authors (see the
citations in Harvey, 1989; West and Harrison, 1997; Durbin and
Koopman, 2001) can only accommodate for regression type state space
models and for local level models. More general structures, such
that of model (\ref{model3}) can only be dealt with via
simulation-based methods, such as Monte Carlo simulation. For
high-dimensional dynamical systems and in particular for observation
covariance estimation, the proposal of PSPP state space model of
Section \ref{s4s4atseries} offers a fast and reliable approximate
estimation procedure, which can be applied for a wide range of time
series.

\subsection{The US investment and business inventory data}\label{data2}

We consider US investment and change in business inventory data,
which are deseasonalised and they are measured quarterly into a
bivariate time series (variable $y_{1t}$: US investment data and
variable $y_{2t}$: US change in inventory data) over the period
1947-1971. The data are fully described and tabulated in L\"utkepohl
(1993) and Reinsel (1997, Appendix A). The data are plotted in
Figure \ref{fig1} with their forecasts, which are generated by
fitting the linear trend PSPP state space model

\begin{figure}[t]
 \epsfig{file=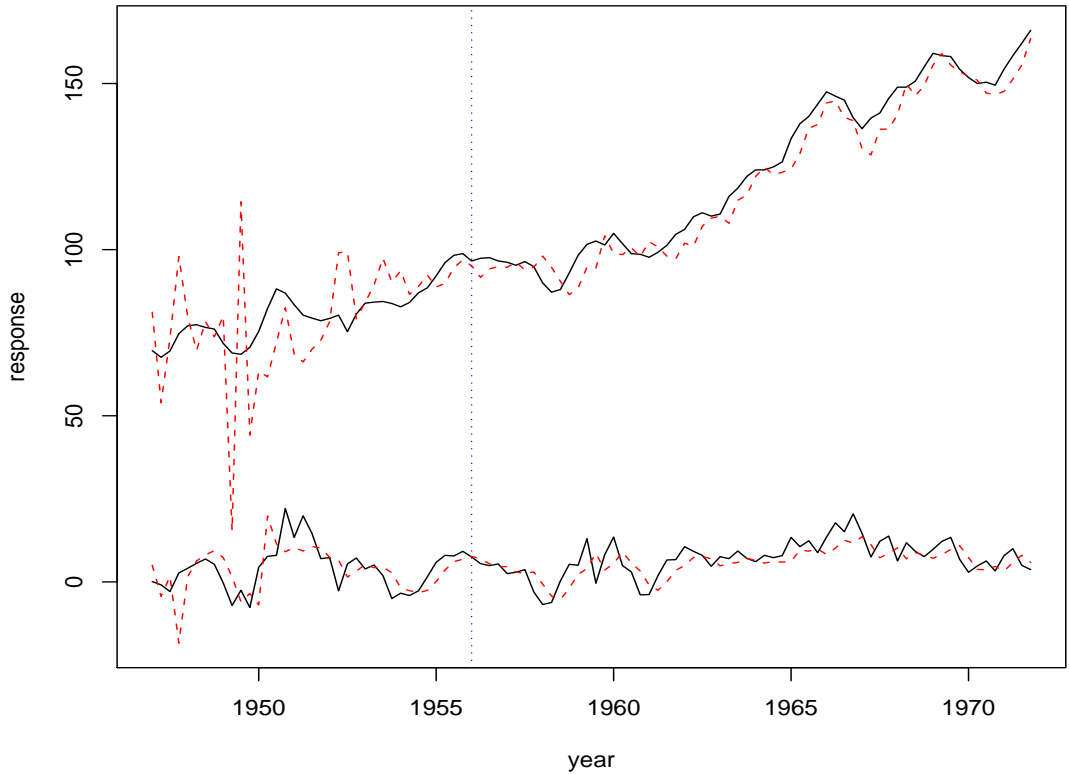, height=12cm, width=15cm}
 \caption{US Investment and Change in Inventory time series $y_t=[y_{1t}~y_{2t}]'$
 with its 1-step forecast mean $f_t=[f_{1t}~f_{2t}]'$. The top solid line
 shows $y_{1t}$ and the bottom solid line shows $y_{2t}$; the top
 dashed line shows $f_{1t}$ and the bottom dashed line shows $f_{2t}$.}\label{fig1}
\end{figure}

\begin{equation}\label{model4}
Y_t = \left[\begin{array}{cc} 1 & 0 \\ 0 & 1 \end{array}\right]
X_t + \epsilon_t,\quad X_t=\left[\begin{array}{cc} 1 & 1 \\ 0 &
1\end{array}\right] X_{t-1} + \omega_t, \quad \epsilon_t\sim
(0,V), \quad \omega_t\sim (0,W_t),
\end{equation}
where here we have not specified the distributions of $\epsilon_t$
and $\omega_t$ as normal and we have replaced the time-invariant $W$
of Section \ref{s4s4atseries} with a time-dependent $W_t$. Model
(\ref{model4}) is a PSPP linear trend state space model, for which
we choose the priors $m_0=[80.622~4.047]'$ (mean of
$[Y_{1t}~Y_{2t}]'$ for $t=1941-1956$, indicated in Figure \ref{fig1}
by the vertical line), $P_0=1000 I_2$ (weakly informative prior
covariance matrix or low precision $P_0^{-1}\approx 0$) and
$$
V_0=\left[\begin{array}{cc} 66.403 & 22.239 \\ 22.239 & 46.547
\end{array}\right],
$$
which is taken as the sample covariance matrix of $Y_{1t}$ and
$Y_{2t}$, for the time period 1941-1955. The covariance matrix $W_t$
measures the durability and the stability of the change or evolution
of the states $X_t$. Here we specify $W_t$ with 2 discount factors,
$\delta_1$ and $\delta_2$, as follows. With $G$ as the evolution
matrix of $X_t$ and $\Delta$ the discount matrix
$$
G=\left[\begin{array}{cc} 1 & 1\\ 0 & 1\end{array}\right], \quad
\Delta=\left[\begin{array}{cc} \delta_1 & 0 \\ 0 &
\delta_2\end{array}\right],
$$
we have
$$
W_t=\Delta^{-1/2} GP_{t-1}G' \Delta^{-1/2} -GP_{t-1}G',
$$
where $R_t$ in the recursions of Section \ref{s4s4atseries} is
replaced by $R_t=GP_{t-1}G'+W_t$. Although this discounting
specification is not advocated by West and Harrison (1997, \S6.4),
it has been successfully used (McKenzie, 1974, 1976; Abraham and
Ledolter, 1983, Chapter 7; Ameen and Harrison, 1985; Goodwin, 1997).

The values of $\delta_1$ and $\delta_2$ are chosen by
experimentation. The above model gave the best result with a
combination of discount factors $\delta_1=0.2$ and $\delta_2=0.4$.
The performance measures were $\textrm{MSSE}=[1.001~1.101]'$,
$\textrm{MSE}=[111.165~66.941]'$, $\textrm{MAE}=[6.718~6.855]'$ and
$\textrm{ME}=[0.076~1.725]'$. Other combinations of $\delta_1$ and
$\delta_2$ yield less accurate results, with the usual effect that
one of the two series $y_{1t}$ and $y_{2t}$ is accurately predicted,
but the other one series is badly predicted. This problem certainly
arises when $\delta_1=\delta_2$, which clearly indicates the need of
multiple discounting. Also, Figure \ref{fig2} plots the observation
variance, covariance and correlation estimates in the time period
1956-1970. From this plot we observe that the variability of the
change in inventory time series component $y_{2t}$ is much larger
than that of $y_{1t}$. The estimate of the observation correlation
indicates the high cross-correlation between the two series.

\begin{figure}[t]
 \epsfig{file=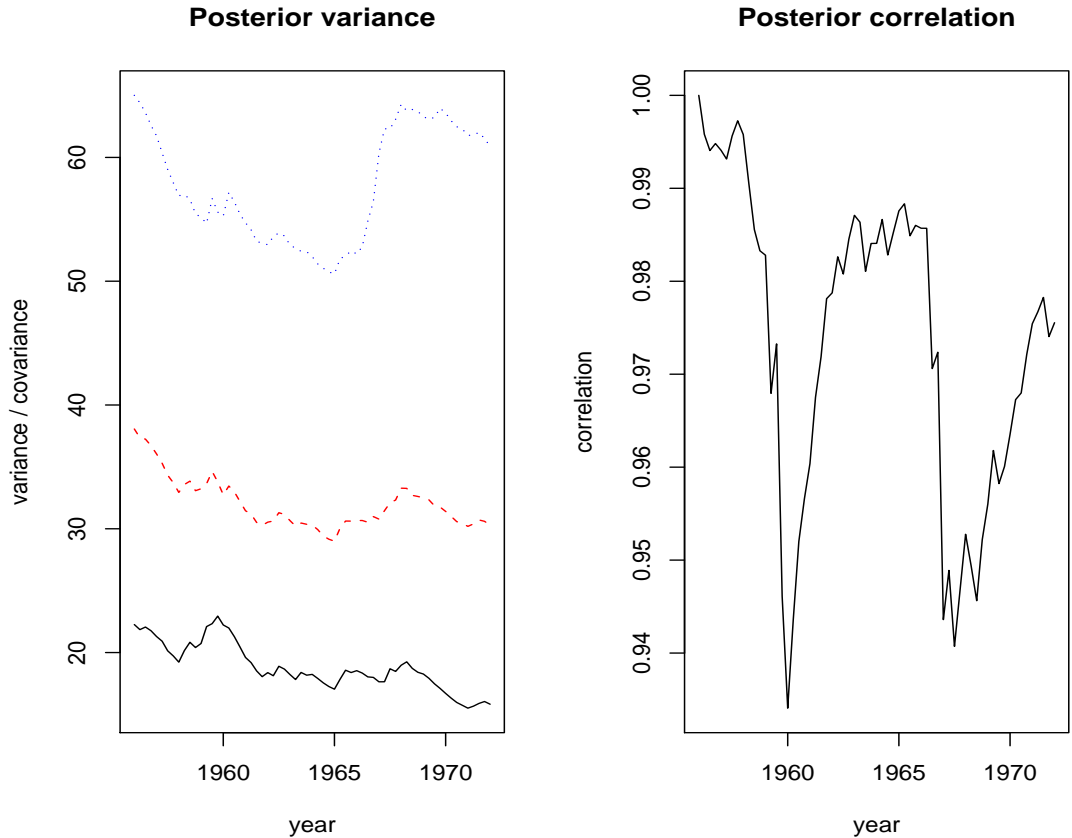, height=12cm, width=15cm}
 \caption{Posterior estimates of the observation covariance matrix $V=(V_{ij})_{i,j=1,2}$ and estimates of the correlation
 $\rho=V_{12}/\sqrt{V_{11}V_{22}}$. In the left panel graph, shown
 are: estimate of the variance $V_{11}$ (solid line), estimate of the variance $V_{12}$ (dashed line), and estimate
 of the variance $V_{22}$ (dotted line). In the right panel graph, the solid line shows the estimate of $\rho$.}\label{fig2}
\end{figure}

\section{Discussion}\label{s5}

This paper develops a method for approximating the first two moments
of the posterior distribution in Bayesian inference. This work is
particularly appealing in regression and time series problems when
the response and parameter distributions are only partially
specified by means and variances. Our partially specified prior
posterior (PSPP) models offer an approximation to prior/posterior
updating, which is appropriate for sequential application, such as
in time series analysis. The similarities and differences with Bayes
linear methods are indicated and, although the authors do believe
that Bayes linear methods offer a great statistical tool, it is
pointed out that in some problems, considered in this paper and in
particular for time series data, the PSPP modelling approach can
offer advantages as opposed to Bayes linear methods.

PSPP models are developed having in mind Bayesian inference for
multivariate state space models when the observation covariance
matrix is unknown and it is subject to estimation. This paper
outlines the deficiency of the existing methods to tackle this
problem and it is shown empirically that, for a class of important
time series data, including local level, linear trend and seasonal
components, PSPP generates much more accurate and reliable posterior
estimators, which are remarkably fast and applicable to a wide range
of time series data. US investment and change in inventory data are
used to illustrate the capabilities of the PSPP state space models.

Given the similarities of the PSPP with Bayes linear methods, it is
believed that the applicability of the PSPP approach goes beyond the
examples considered in this paper. For example one area that is only
slightly touched, is inference for data following non-normal
distributions, other than the multivariate $t$, the inverted
multivariate $t$, and the Wishart distributions. In this sense a
more detailed comparison of PSPP with Bayes linear methods and in
particular with Bayes linear kinematics (Goldstein and Shaw, 2004),
should shed more light on the performance of PSPP. It is our purpose
to consider such comparisons in a future paper.

\section*{Acknowledgements}

The authors are grateful to the Statistics Department at Warwick
University, where this work was initiated. We are grateful to three
referees for providing helpful comments.

\section*{Appendix}

\begin{proof}[Proof of Theorem \ref{th1}]

$(\Longrightarrow)$ By hypothesis
$\E(X|Y)=\mu_x+A_{xy}(Y-\mu_y)\Rightarrow
\E(X-A_{xy}Y|Y)=\mu_x-A_{xy}\mu_y=$ constant. Furthermore
$\var(X|Y)=E\{(X-\mu_x-A_{xy}(Y-\mu_y))(X-\mu_x-A_{xy}(Y-\mu_y))'|Y\}=
\Sigma_x-A_{xy}\Sigma_yA_{xy}'=$ constant $\Rightarrow
\var(X-A_{xy}Y|Y)=\var(X|Y)=$ constant. It follows that
$X-A_{xy}Y\bot_2Y$.

$(\Longleftarrow)$ The assumption $X-A_{xy}Y\bot_2 Y$ implies that
$\E(X-A_{xy}Y|Y)=\mu$ constant $\Rightarrow \E(X|Y)=A_{xy}Y+\mu$,
which is a linear function of $Y$. Given that $\E(X|Y)$ minimizes
the quadratic prior expected risk and $\mu_x+A_{xy}(Y-\mu_y)$
minimizes this risk among all linear estimators, it follows that
$\E(X|Y)=\mu_x+A_{xy}(Y-\mu_y)$.
\end{proof}

\end{document}